\documentclass[twocolumn,showpacs,preprintnumbers,amsmath,amssymb]{revtex4-2}

\usepackage{graphicx}
\usepackage{dcolumn}
\usepackage{bm}
\usepackage{amsthm}
\usepackage{hyperref}
\usepackage{mathtools}
\usepackage{amssymb}

\newcommand{\cptp}{\mathcal{E}}
\newcommand{\iden}{\mathbb{I}}
\newcommand{\qfi}{\mathcal F}
\newcommand{\PEqfi}{{\mathcal F}_H^{\mathrm{PE}}}
\newcommand{\PDqfi}{\mathcal F_{\theta}^{\mathrm{PD}}}
\newcommand{\tr}{\mathrm{tr}}
\newcommand{\pdone}{\mathcal{E}^{\mathrm{PD},1}_\theta}
\newcommand{\pd}{\mathcal{E}^{\mathrm{PD}}_{\theta}}
\newcommand{\cd}{\Delta}

\newcommand{\HDP}{\mathrm{HDP}}
\newcommand{\SHP}{\mathrm{SHP}}
\newcommand{\DIO}{\mathrm{DIO}}
\newcommand{\SIO}{\mathrm{SIO}}
\newcommand{\hdp}{\mathcal{E}^{\textrm{HDP}}}
\newcommand{\shp}{\mathcal{E}^{\textrm{SHP}}}
\newcommand{\NS}{S^n}

\renewcommand{\Re}{\operatorname{Re}}

\theoremstyle{definition}

\newtheorem{definition}{Definition}
\newtheorem{proposition}{Proposition}
\newtheorem{lemma}{Lemma}
\newtheorem{corollary}{Corollary}

\begin{document}

\preprint{}

\title{Resource theory of dephasing estimation in multiqubit systems}
\author{Zishi Chen}

\author{Xueyuan Hu}
\email{xyhu@sdu.edu.cn}
\affiliation{School of Information Science and Engineering, Shandong University, Qingdao 266237, China}


\begin{abstract}
We present a resource theory to investigate the power of a multqubit system as a probe in the task of dephasing estimation. 
Our approach employs the quantum Fisher information about the dephasing parameter as the resource measure. 
Based on the monotonicity of quantum Fisher information, we propose two sets of free operations in our resource theory, the Hamming distance preserving operations and the selectively Hamming distance preserving operations. 
We derive a necessary condition for the state transformation under these free operations and demonstrate that uniform superposition states are the golden states in our resource theory. 
We further compare our resource theory with the resource theory of coherence and thoroughly investigate the relation between their free operations in both single-qubit and multiqubit cases. 
Additionally, for multiqubit systems, we discover the incompatibility between the resource theory of dephasing estimation and that of $U(1)$ asymmetry, which is responsible for phase estimation.
The condition for enhancing  the performance of a probe state in phase estimation while preserving its ability in dephasing estimation is also discussed.
Our results provide new insights into quantum parameter estimation by the resource-theoretic approach. 
\end{abstract}

\maketitle
\section{INTRODUCTION}
In quantum resource theories \cite{RevModPhys.91.025001, PhysRevLett.119.140402, PhysRevA.90.014102,2022arXiv2212.11105, PhysRevLett.126.090401}, certain quantum properties of physical systems are treated as resources. 
By the resource-theoretic approach, the advantages of these quantum resources are quantified for various quantum information processing tasks such as quantum teleportation \cite{PhysRevA.106.012433}, channel discrimination \cite{PhysRevA.93.042107, PhysRevLett.122.140402} and quantum key distribution \cite{PhysRevA.102.012617}. 
Moreover, the quantum resource theory provides a general framework to characterize laws for state transformations under certain classes of restricted operations called free operations, and thus sheds new light on the research field of fundamental physics, including thermodynamics \cite{Lostaglio_2019}.

Quantum parameter estimation is a task which employs quantum properties to make precise estimation of given parameters \cite{Liu_2020, Toth_2014, RevModPhys.90.035006, doi:10.1116/1.5119961,Lu2015}. For instance, in phase estimation, maximally entangled probe states can improve the error scaling of the estimator \cite{Toth_2014, PhysRevLett.96.010401}. This improvement inspires researchers to investigate the role of quantum properties in quantum parameter estimation \cite{RevModPhys.90.035005}. 
It has been found that there are quantum enhanced estimation tasks without entanglement
\cite{PhysRevA.81.022108, RevModPhys.90.035006}, indicating that entanglement is not the unique resource which underlies the quantum advantage in parameter estimation. In particular, for the estimation of noise in teleportation-covariant channels, including phase damping (PD) channels \cite{nielsen2002quantum,PhysRevA.94.052108}, entanglement is not a necessary condition for the optimal probe states \cite{RevModPhys.90.035006, Kolodynski_2013, PhysRevLett.118.100502, Pirandola2018}. 

Recently, the notions of quantum parameter estimation have been introduced to the framework of resource theories and made significant progresses. The quantum Fisher information (QFI), which represents the precision limit of probe states \cite{PhysRevLett.72.3439}, can be used as a resource measure of coherence and asymmetry \cite{PhysRevA.99.012329, PhysRevA.90.014102, PhysRevA.93.022122, PhysRevLett.127.200402, 2022arXiv220503245K}, due to its nice mathematical properties such as convexity and monotonicity \cite{Liu_2020, PhysRevA.91.042104}. 
Furthermore, by the resource-theoretic approach, an operational interpretation of the QFI is proposed in quantum thermodynamics \cite{PhysRevLett.129.190502}. 
In this work, we propose and give answers to the following questions: What is the nature of the quantum resource for dephasing estimation? How does this resource relate to some well-known resources, such as quantum coherence and $U(1)$ asymmetry?

We establish a resource theory of dephasing estimation in multiqubit systems. We first prove that incoherent states are free states in this present resource theory. Due to the monotonicity of QFI, we define two sets of free operations in the resource theory of dephasing estimation, the Hamming distance preserving operations (HDP) and the selectively Hamming distance preserving operations (SHP). Then we derive a necessary condition for the transformed states under HDP. Furthermore, uniform superposition states can be transformed into any state under SHP, that is, uniform superposition states are golden states in the resource theory of dephasing estimation. Moreover, by comparing our resource theory and the resource theory of quantum coherence, we show the relation between their free operations in both single-qubit and multiqubit cases. Finally, we consider the resource theory of $U(1)$ asymmetry. By employing SHP, it is possible to improve the performance of the probe state in phase estimation, while also preserving the ability of dephasing estimation.
\section{NOTIONS AND PRELIMINARIES}
\label{PRELIMINARIES}
We denote the Hilbert space associated with the qubits of our consideration by $\mathcal{H}$. Let $\mathcal{L}(\mathcal{H})$ be the set of liner operators on $\mathcal{H}$. Any quantum operation on a physical system can be mathematically represented by a completely-positive and trace-preserving (CPTP) map $\mathcal{E}:\mathcal{L}(\mathcal{H})\rightarrow\mathcal{L}(\mathcal{H})$. In this paper, we focus on the situations where the input and output spaces are the same. We use $\iden$ as the identity matrix of the input (output) space. In the following, we briefly review the definitions of the Fisher information (FI), the QFI, and the PD channels \cite{Liu_2020, nielsen2002quantum}.

\subsection{FISHER INFORMATION (FI) AND QUANTUM FISHER INFORMATION (QFI)}
In a general parameter estimation task, a probe system, initially in state $\rho$, is sent through a parameterized quantum channel $\mathcal{E}_\theta$. The output state $\rho(\theta)\equiv\mathcal{E}_\theta(\rho)$ is then measured by a positive-operator-valued measurement (POVM)  $\{M_x\}$, where $M_x$ is the positive semi-definite operator satisfying $\sum_{x}M_x=\iden$, and each measurement result $x$ is obtained with the probability $p(x|\theta)=\tr(\rho(\theta) M_x)$.
Finally, the parameter $\theta$ is estimated based on the probability distribution of measurement results.

The Cram\'{e}r-Rao bound provides a lower bound for the variance of an unbiased estimator \cite{Liu_2020, RevModPhys.90.035006, Toth_2014}
\begin{equation}
	{\rm Var}(\theta) \ge \dfrac{1}{F(\theta)} \ge \dfrac{1}{\qfi(\rho(\theta))}.
\end{equation}
Here the Fisher information $F(\theta)$ is defined as \cite{Liu_2020, RevModPhys.90.035006, PhysRevLett.127.200402}
\begin{equation}
	F(\theta)=\sum_{x}\dfrac{(\partial p(x|\theta))^2}{p(x|\theta)},
\end{equation}
where $\partial \coloneqq \partial/\partial\theta$. It depends on both the probe state $\rho$ and the measurement $\{M_x\}$.
The QFI $\qfi(\rho(\theta))$ is defined as \cite{Liu_2020, RevModPhys.90.035006, Toth_2014, PhysRevLett.127.200402}
\begin{equation}
	\qfi(\rho(\theta))=\tr(\rho(\theta) L^2),
\end{equation}
where $L$ is the symmetric logarithmic derivative (SLD) and defined by $\frac12(L\rho(\theta)+\rho(\theta) L)=\partial \rho(\theta)$. The QFI satisfies the following properties \cite{Liu_2020, RevModPhys.90.035006}.\\
(P1) Additivity. $\qfi(\rho(\theta)\otimes\rho'(\theta))=\qfi(\rho(\theta))+\qfi(\rho'(\theta))$. It implies $\qfi(\rho(\theta)^{\otimes n})=n\qfi(\rho(\theta))$.\\
(P2) Convexity. $\qfi(\sum_jp_j\rho_j(\theta))\leq\sum_jp_j\qfi(\rho_j(\theta))$.\\
(P3) Monotonicity. $\qfi(\rho(\theta))\geq\qfi(\cptp(\rho(\theta)))$, for any CPTP map $\cptp$ which is independent of $\theta$.\\

\subsection{PHASE DAMPING CHANNELS}

In this paper, we focus on the phase damping (PD) channel for single-qubit and multiqubit states, which is denoted by $\pd$. 

A single-qubit state $\sigma=\frac12(\iden+r\cos{\phi}\sigma_x+r\sin{\phi}\sigma_y+z\sigma_z)$, where $\sigma_x$, $\sigma_y$ and $\sigma_z$ are two-dimensional Pauli matrices, can be represented as a Bloch vector $(r\cos{\phi},r\sin{\phi},z)$, where $-1\le z \le 1$, $0\le r \le 1$ and $0\le \phi < 2\pi$. Then the density matrix of the state $\sigma$ can be written as
\begin{equation}
	\label{one_qubit_density}
	\sigma=\dfrac{1}{2}\begin{bmatrix}
		1+z&re^{-i\phi}\\
		re^{i\phi}&1-z
	\end{bmatrix}.
\end{equation}
After the action of a PD channel \cite{nielsen2002quantum}, the state becomes
\begin{equation}
	\pdone(\sigma)=\dfrac{1}{2}\begin{bmatrix}
		1+z&re^{-\theta-i\phi}\\
		re^{-\theta+i\phi}&1-z
	\end{bmatrix},
\end{equation}
where $\theta$ is the parameter of the PD channel. 

Now consider a general $n$-qubit state $\rho=\sum_{x=0,y=0}^{2^n-1,2^n-1}\rho_{xy}|x\rangle \langle y|$, where $x$ and $y$ are binary strings of length $n$ and $\rho_{xy}=\langle x|\rho|y\rangle$. Let each qubit transmit through a PD channel $\pdone$, and the output $n$-qubit state reads
\begin{equation}\label{eq:pd}
	\pd(\rho)=\sum_{x=0,y=0}^{2^n-1,2^n-1}e^{-h(x,y)\theta}\rho_{xy}|x\rangle \langle y|,\ n=1,2,...
\end{equation}
where $h(i,j)$ is the Hamming distance between two binary strings $i$ and $j$. 
The mathematical definition of the Hamming distance is $h(i,j)=\sum_{m=1}^{n}i_m\oplus j_m$, where $i_m(j_m)$ represents the $m$th bit of the string $i(j)$. The task, which is to estimate $\theta$ by a probe state, is called dephasing estimation.

\section{STRUCTURE OF THE RESOURCE THEORY OF DEPHASING ESTIMATION}
\label{STRUCTURE}
In the resource theory of dephasing estimation, we use the QFI as a resource measure. 
Precisely, we define the power of a probe state $\rho$ in dephasing estimation as follows
\begin{equation}
	\PDqfi(\rho):=\qfi(\pd(\rho)).
\end{equation}

The free states are those which do not have the ability of dephasing estimation, i.e., $\rho$ is a free state if and only if $\PDqfi(\rho)=0$. 
Apparently, an incoherent state $\rho_{\mathrm{inc}}$ on computational basis belongs to the set of free states, because $\pd(\rho_{\mathrm{inc}})$ is independent of $\theta$.
In the following, we prove that any state which is not incoherent is a resourceful state.

Assume that there is at least one non-diagonal element $\rho_{xy}=\left|\rho_{xy} \right|e^{-i\beta} \neq0$ in the probe state $\rho_{\mathrm{coh}}$. After the action of $\pd$, one implements the following measurement
\begin{equation}
	\begin{split}
		&M_1=(|x\rangle\langle x|+e^{-i\beta}|x\rangle\langle y|+e^{i\beta}|y\rangle\langle x|+|y\rangle\langle y|)/2,\\
		&M_2=(|x\rangle\langle x|-e^{-i\beta}|x\rangle\langle y|-e^{i\beta}|y\rangle\langle x|+|y\rangle\langle y|)/2,\\
		&M_3=\iden-|x\rangle\langle x|-|y\rangle\langle y|,
	\end{split}
\end{equation}
and obtains the outcome $a\in\{1,2,3\}$ with probability $p(a|\theta)=\tr(\pd(\rho_{\mathrm{coh}}) M_a)$. Direct calculation leads to
\begin{equation}
	\begin{split}
		p(1|\theta)&=\dfrac{1}{2}(\rho_{xx}+\rho_{yy}+2\left|\rho_{xy} \right|e^{-h(x,y)\theta}),\\
		p(2|\theta)&=\dfrac{1}{2}(\rho_{xx}+\rho_{yy}-2\left|\rho_{xy} \right|e^{-h(x,y)\theta}),\\
		p(3|\theta)&=1-p(1|\theta)-p(2|\theta).
	\end{split}
\end{equation}
It follows that
\begin{equation}
	\begin{split}
		F(\theta)&=\sum_{a=1}^{3}\dfrac{(\partial p(a|\theta))^2}{p(a|\theta)}\\
		&=\dfrac{[h(x,y)\left|\rho_{xy} \right|e^{-h(x,y)\theta}]^2}{p(1|\theta)}\\&+\dfrac{[h(x,y)\left|\rho_{xy} \right|e^{-h(x,y)\theta}]^2}{p(2|\theta)}>0.
	\end{split}
\end{equation}
Because the QFI is lower bounded by the FI, we have $\PDqfi(\rho_{\mathrm{coh}})\geq F(\theta)>0$. Therefore, a free state in the resource theory of dephasing estimation is of the following form
\begin{equation}
    \rho_{\mathrm{free}}=\sum_{x=0}^{2^n-1}\rho_x |x\rangle\langle x|.
\end{equation}

As for free operations, the minimum requirement is that they do not increase the resource of states. Nevertheless, considering the complexity in calculating the QFI, defining the free operations to be the whole set of $\PDqfi$-nonincreasing operations would cause difficulties in analysing the properties of the present resource and generalizing the method of studying the resources in dephasing estimation. With this in mind, we propose two sets of free operations in the following.

The first set of free operations $\mathbb{F}_1$ is defined as the set of all CPTP maps commutative with any PD channel,
\begin{equation}
    \cptp\circ\pd=\pd\circ\cptp,\ \forall \cptp\in\mathbb{F}_1,\label{eq:F1}
\end{equation}
and will be thoroughly studied in Sec. \ref{subsection:HDP}. The second set of free operations $\mathbb{F}_2$, studied in Sec. \ref{subsection:SHP}, is defined as those with a Kraus decomposition $\cptp=\sum_j\mathcal{K}_j$ (with $\mathcal{K}_j(\cdot)=K_j(\cdot) K_j^\dagger$ and $K_j$ being the Kraus operators satisfying $\sum_j K_j^\dagger K_j=\iden$), such that each Kraus branch $\mathcal{K}_j$ can commute with PD channels, i.e.,
\begin{equation}
    \mathcal{K}_j\circ\pd=\pd\circ\mathcal{K}_j,\ \forall j,\  \forall \cptp\in\mathbb{F}_2.\label{eq:F2}
\end{equation}
By definition, $\mathbb{F}_2\subseteq \mathbb{F}_1$.
The detailed comparison between these two sets will be discussed in Sec. \ref{subsection:HDP_SHP_relation}.

The monotonicity of $\PDqfi$ under these free operations can be proved by employing the monotonicity (P3) of the QFI. Precisely, if $\cptp\in\mathbb{F}_1$, we have
\begin{equation}
    \PDqfi(\cptp(\rho))=\qfi(\cptp(\pd(\rho)))\leq\qfi(\pd(\rho))=\PDqfi(\rho),
\end{equation}
for any $n$-qubit state $\rho$. 

Also, it is worth noting that, even in the single-qubit case, $\mathbb{F}_1$ is a strict subset of the $\PDqfi$-nonincreasing operations. See \textbf{Appendix \ref{sec:appendix-a_0}} for detailed discussion.

\subsection{HAMMING DISTANCE PRESERVING OPERATIONS (HDP)}\label{subsection:HDP}
In this subsetion, we first give the definition of Hamming distance preserving operations (HDP) and then prove that HDP is equivalent to $\mathbb{F}_1$ defined in Eq.~(\ref{eq:F1}). Moreover, we derive a necessary condition for state transformations under HDP.

\begin{definition}
	A CPTP map $\hdp:\mathcal{L}(\mathcal{H}^n)\mapsto\mathcal{L}(\mathcal{H}^n)$ belongs to the set of Hamming distance preserving operations (HDP) if and only if
	\begin{equation}
		\label{D_HDP}
		\langle i|\hdp(|x\rangle\langle y|)|j\rangle=0,
	\end{equation}
	for all $i,j,x,y\in\NS$ satisfying $h(i,j)\neq h(x,y)$, where $\NS$ is the set of $n$-bit strings.
\end{definition}

The following proposition shows the equivalence between the Hamming distance preserving condition and the condition of commutativity with PD channels for completely positive maps.

\begin{proposition}\label{prop:1}
	Let $\Tilde{\cptp}$ be a completely positive (CP) map. Then, 
 \begin{equation}\label{eq:commu_pd}
     \Tilde{\cptp}\circ\pd=\pd\circ\Tilde{\cptp},
 \end{equation}
 if and only if
 	\begin{equation}
		\label{eq:hdp}
		\langle i|\Tilde{\cptp}(|x\rangle\langle y|)|j\rangle=0,
	\end{equation}
 for all $i,j,x,y\in\NS$ satisfying $h(i,j)\neq h(x,y)$.
 \end{proposition}
\begin{proof}
 From Eq.~(\ref{eq:pd}), Eq.~(\ref{eq:commu_pd}) can be rewritten as
	\begin{equation}
		\label{HDP_P}
		(e^{-h(i,j)\theta}-e^{-h(x,y)\theta})\langle i|\Tilde{\cptp}(|x\rangle\langle y|)|j\rangle=0\quad\forall i,j,x,y.
	\end{equation}
	Equivalently, $\langle i|\Tilde{\cptp}(|x\rangle\langle y|)|j\rangle=0$ if $h(x,y)\neq h(i,j)$.	
\end{proof}

By setting $\Tilde{\cptp}$ in the above proposition to be a CPTP map $\cptp$, we directly have $\mathbb{F}_1=\HDP$. Therefore, we will label the first set of free operations defined in Eq.~(\ref{eq:F1}) as HDP.

The following proposition characterizes a necessary condition of state transformations under HDP.

\begin{proposition}\label{prop:2}
	For any $n$-qubit state $\rho$, it holds that
	\begin{equation}
		\label{HDP_bound}
		\lvert \langle i|\hdp(\rho)|j\rangle \rvert
  \le {\sum_{x,y}}^\prime \lvert \rho_{xy}\rvert\sqrt{p_{i|x}p_{j|y}},
	\end{equation}
	where $\hdp\in\HDP$, $p_{i|x}=\langle i|\hdp(|x\rangle\langle x|)|i\rangle$, and $\sum^\prime_{x,y}$ stands for a summation over all $x,y$ which satisfy $h(x,y)=h(i,j)$. 
\end{proposition}
\begin{proof}
Based on Eq.~(\ref{D_HDP}),  each element of the output density matrix is given by
	\begin{equation}\label{eq:hdp_bound1}
		 \langle i|\hdp(\rho)|j\rangle ={\sum_{x,y}}^\prime\rho_{xy}\langle i|\hdp(|x\rangle\langle y|)|j\rangle.
	\end{equation}
 Now write $\hdp$ in a Kraus decomposition form $\hdp(\cdot)=\sum_a D_a(\cdot)D_a^\dagger$ and define the vector $\Vec{X}^{xy}=(X^{xy}_1, X^{xy}_2,\dots)$ with $X^{xy}_a\equiv\langle x|D_a|y\rangle$. Then we have $\langle i|\hdp(|x\rangle\langle y|)|j\rangle=\Vec{X}^{ix}
 \cdot \Vec{X}^{jy*}$, which implies $p_{i|x}=|\Vec{X}^{ix}|^2$. It follows from the Cauchy-Schwarz inequality that
 \begin{equation}
     \left| \langle i|\hdp(|x\rangle\langle y|)|j\rangle\right|\leq   \sqrt{p_{i|x}p_{j|y}}.
 \end{equation}
 The equality sign holds for one of the two cases: (1) either $p_{i|x}$ or $p_{j|y}$ vanishes; (2) a parameter $\xi$ exists such that $\Vec{X}^{ix}=\xi \Vec{X}^{jy}$.
Therefore, 
\begin{eqnarray}
\left| \langle i|\hdp(\rho)|j\rangle\right| &\leq & {\sum_{x,y}}^\prime |\rho_{xy}| |\langle i|\hdp(|x\rangle\langle y|)|j\rangle|\nonumber\\
& \leq & {\sum_{x,y}}^\prime |\rho_{xy}|\sqrt{p_{i|x}p_{j|y}}.
\end{eqnarray}
The equality sign in the first line holds if the terms in the summation of Eq.~(\ref{eq:hdp_bound1}) have the same phase.
\end{proof}

If we focus on the diagonal elements of the output state, the condition $h(x,y)=h(i,i)=0$ leads to $x=y$, and Eq.~(\ref{eq:hdp_bound1}) reduces to
\begin{equation}
\label{diag:hdp_bound}
    \langle i|\hdp(\rho)|i\rangle = \sum_x\rho_{xx}p_{i|x}.
\end{equation}
It means that the diagonal elements of the output state are independent of the off-diagonal elements of the input state.

Moreover, an $n$-qubit state $\rho$ can be written as
\begin{equation}
    \rho=\sum_{h=0}^n \rho^{(h)},\ \rho^{(h)}\equiv\sum_{x,y:h(x,y)=h}\rho_{xy}|x\rangle\langle y|.
\end{equation}
We call $\rho^{(h)}$ a Hamming mode of $\rho$. Then, \textbf{Proposition \ref{prop:2}} indicates that each Hamming mode of the input state is independently mapped by HDP to the corresponding Hamming mode of the output state, namely,
\begin{equation}
    \hdp(\rho^{(h)})=[\hdp(\rho)]^{(h)}.
\end{equation}
It follows that, if $\rho^{(h)}=0$ for some $h$, then $[\hdp(\rho)]^{(h)}=0$. For example, the state $|\psi_1\rangle=\frac{1}{\sqrt{2}}(|0\rangle+|1\rangle)$ cannot be transformed by HDP to state $|\psi_2\rangle=\frac{1}{\sqrt{2}}(|0\rangle+|3\rangle)$, because $(|\psi_1\rangle\langle\psi_1|)^{(2)}=0$ but $(|\psi_2\rangle\langle\psi_2|)^{(2)}=\frac{1}{2}(|0\rangle\langle 3|+|3\rangle\langle 0|)\neq 0$. This sets the essential difference between the present resource theory and the resource theory of coherence, because $|\psi_1\rangle$ and $|\psi_2\rangle$ can be transformed to each other by incoherent unitary operators. More details of the comparison between the two resource theories are to be illustrated in Sec. \ref{COMPARISON}.

An immediate problem is then whether the bound in Eq. (\ref{HDP_bound}) can be reached. We will prove the attainability in single-qubit cases, and show that it cannot be reached for general $n$-qubit input states. Further, we will characterize a set of $n$-qubit input states, for which the bound in Eq. (\ref{HDP_bound}) can be reached.

The set of states, which can be obtained from a given state $\rho$ by HDP, is called the HDP cone of $\rho$. By \textbf{Propostion \ref{prop:2}}, the HDP cone can be obtained for any single-qubit state. 

\begin{corollary}
    For two single-qubit states with Bloch vectors $\vec r=(r\cos{\phi},r\sin{\phi},z)$ and $\vec{r}'=(r'\cos{\phi'},r'\sin{\phi'},z')$, $\vec r$ can be transformed to $\vec{r}'$ via HDP if and only if
    \begin{equation}
		\label{SIO_region}
		r'\leq\left\{
		\begin{aligned}
			& r\sqrt{\frac{1-z^{\prime2}}{1-z^2}}, & |z'|\geq|z|\\
			& r,  & |z'|< |z| .
		\end{aligned}
		\right.
	\end{equation}
 \end{corollary}
\begin{proof}
If $\vec{r}'$ can be obtained from $\vec r$ via HDP, Eq.~(\ref{HDP_bound}) gives
    \begin{eqnarray}
    \label{r_bound}
      r'&\leq         &r(\sqrt{p_{0|0}p_{1|1}}+\sqrt{p_{0|1}p_{1|0}})\nonumber\\
     &=&r\cos(\theta_0-\theta_1),\label{eq:bound_qubit}
    \end{eqnarray}
    where we set $p_{0|0}=\cos^2\theta_0$, $p_{1|1}=\cos^2\theta_1$, and $\theta_0,\theta_1\in[0,\frac{\pi}{2}]$.
Furthermore, from Eq.~(\ref{diag:hdp_bound}), we have
     \begin{equation}\label{eq:zp}
           \dfrac{1+z'}{2}=\cos^2\theta_0\dfrac{1+z}{2}+(1-\cos^2\theta_1)\dfrac{1-z}{2}.
    \end{equation} 
It follows that
    \begin{equation}\label{eq:bound_z1}
    \begin{split}
        z'^2&=[\dfrac{1}{2}\cos2\theta_0(1+z)-\dfrac{1}{2}\cos2\theta_1(1-z)]^2\\
        &=\dfrac{1}{4}(\Re[e^{2i\theta_0}(1+z)-e^{2i\theta_1}(1-z)])^2\\
        &\le\dfrac{1}{4}|e^{2i\theta_0}(1+z)-e^{2i\theta_1}(1-z)|^2\\
        &=1-(1-z^2)\cos^2(\theta_0-\theta_1),
    \end{split}
    \end{equation} 
or equivalently,
\begin{equation}\label{eq:bound_z}
    \cos^2(\theta_0-\theta_1)\leq\min\{1,\frac{1-z^{\prime2}}{1-z^2}\}.
\end{equation}
By combining Eqs. (\ref{r_bound}) and (\ref{eq:bound_z}), we conclude that Eq. (\ref{SIO_region}) is a necessary condition for the transformation from $\vec r$ to $\vec{r}'$.

The sufficiency of Eq. (\ref{SIO_region}) is proved by constructing a channel $\hdp_*\in\HDP$ such that the boundary states satisfying the equality sign in Eq. (\ref{SIO_region}) can be reached. Precisely, we have $\hdp_*(\cdot)=K_1(\cdot) K_1^\dagger+K_2(\cdot) K_2^\dagger$ with
    \begin{eqnarray}\label{eq:Extreme_Kraus_qubit}
        K_1=\left(\begin{array}{cc}
        \cos\theta_0 & 0 \\
        0 & e^{i(\phi-\phi')}\cos\theta_1
     \end{array}\right),\nonumber\\
     K_2=\left(\begin{array}{cc}
        0 & e^{i(\phi-\phi')}\sin\theta_1 \\
        \sin\theta_0 & 0
        \end{array}\right),
        \end{eqnarray}
where $\theta_0,\theta_1\in[0,\frac{\pi}{2}]$ are determined by the parameters in the input and output states as follows. 

If $|z'|< |z|$, we choose $\sin{2\theta_0}=\sin{2\theta_1}=\frac{z'}{z}$. It can be checked that the channel $\hdp_*$ transforms the state $\vec r$ to $\vec{r}'$ with $r'=r$ and $z'=z\sin{2\theta_0}$.

If $|z'|\geq|z|$, the equality sign in Eq.(\ref{eq:bound_z1}) holds when
\begin{equation}
\sin 2\theta_0(1+z)-\sin 2\theta_1(1-z)=0.
\end{equation}
Combining this equation with Eq. (\ref{eq:zp}), we obtain
\begin{equation}
    \cos2\theta_0=\frac{z+z^{\prime2}}{z'(1+z)},\cos2\theta_1=\frac{z-z^{\prime2}}{z'(1-z)}.
\end{equation}
The channel $\hdp_*$ with the above parameters gives an output state with $z'=\sqrt{1-(1-z^2)\cos^2(\theta_0-\theta_1)}$ and $r'=r\cos(\theta_0-\theta_1)$.

\end{proof}
\begin{figure}
	\centering
	\includegraphics[width=0.45\textwidth]{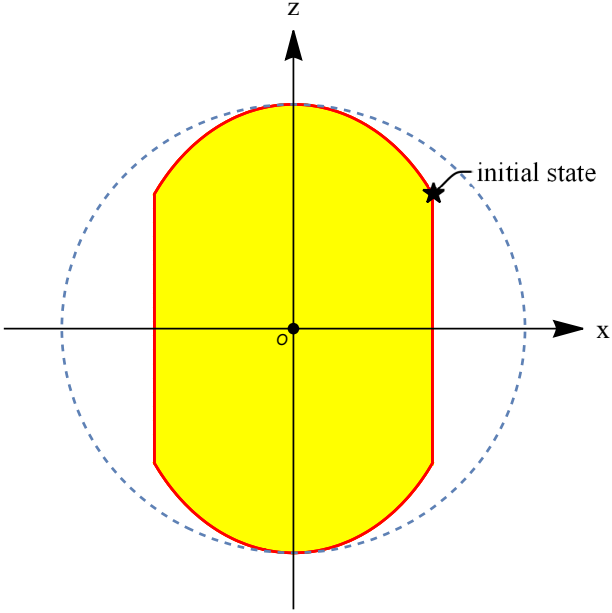}
	\caption{The colored region is the projection of the SHP cone in the x-z plane, when the initial Bloch vector is $(0.6,0,0.6)$. The initial state is marked by the star. The area bounded by the dotted line contains all single-qubit states.}
	\label{f_1}
\end{figure}

The unitary operator $U(\phi)=\textrm{diag}(1,e^{-i\phi})$ belongs to HDP. Based on these unitary operations, the HDP cone is rotationally symmetric about the $z$ axis. That is, the Bloch vector $(r',0,z')$ can be transformed to Bloch vector $(r'\cos{\phi'},r'\sin{\phi'},z')$. 
In FIG.$\ref{f_1}$, we plot the projection of the SHP cone in the x-z plane for a given initial state, whose Bloch vector is  $(0.6,0,0.6)$.

In the above proof, we show that the bound in Eq. (\ref{HDP_bound}) can be reached for single-qubit input states. Nevertheless, we will show in the following that this bound cannot be reached in general. A counterexample goes as follows. Consider a $2$-qubit state in the following form
 \begin{equation}
     \rho=\frac14\left(\begin{array}{cccc}
         1 & \frac{1}{\sqrt{2}} & \frac{1}{\sqrt{2}} & 0 \\
        \frac{1}{\sqrt{2}} & 1 & 0 & 0 \\
        \frac{1}{\sqrt{2}} & 0 & 1 & 0 \\
        0 & 0 & 0 & 1
     \end{array}\right).
 \end{equation}
  We aim to maximise $|\langle0|\cptp(\rho)|1\rangle|$ for $\cptp\in\HDP$. The bound in Eq. (\ref{HDP_bound}) gives
 \begin{eqnarray}
     |\langle0|\cptp(\rho)|1\rangle|& \leq &\sqrt{p_{0|0}p_{1|1}}|\rho_{01}| + \sqrt{p_{0|1}p_{1|0}}|\rho_{10}|\nonumber\\
     & & +\sqrt{p_{0|0}p_{1|2}}|\rho_{02}| + \sqrt{p_{0|2}p_{1|0}}|\rho_{20}|\nonumber\\
     &\leq & |\rho_{01}|+|\rho_{02}|.\label{eq:bound1}
 \end{eqnarray}
 The second equality sign holds if and only if $p_{0|0}=p_{1|1}=p_{1|2}\equiv \cos^2\phi$ and $p_{1|0}=p_{0|1}=p_{0|2}\equiv \sin^2\phi$.

 Now suppose $\cos^2\phi\neq0$. It follows that the first equality sign in Eq.~(\ref{eq:bound1}) holds only if $\Vec{X}^{00}=e^{i\eta_1}\Vec{X}^{11}=e^{i\eta_2}\Vec{X}^{12}$, which in turn gives $\cos^2\phi=|\Vec{X}^{11}\cdot\Vec{X}^{12*}|=|\langle 1|\cptp(|1\rangle\langle 2|)|1\rangle|$. However, the rhs of the above equation equals zero from the Hamming distance preserving condition. This violates the assumption that $\cos^2\phi\neq0$. Therefore, we have $\sin^2\phi=1$. Then the first equality sign in Eq. (\ref{eq:bound1}) holds only if $\Vec{X}^{10}=e^{i\xi_1}\Vec{X}^{01}=e^{i\xi_2}\Vec{X}^{02}$, which in turn gives $\sin^2\phi=|\Vec{X}^{01}\cdot\Vec{X}^{02*}|=\langle 0|\cptp(|1\rangle\langle 2|)|0\rangle$, but again the rhs equals zero because of the Hamming distance preserving condition. Therefore, it is impossible to transform $\rho$ into 
 \begin{equation}
     \rho'=\frac14\left(\begin{array}{cccc}
         1 & \sqrt2 & 0 & 0 \\
        \sqrt2 & 2 & 0 & 0 \\
        0 & 0 & 0 & 0 \\
        0 & 0 & 0 & 1
     \end{array}\right),
 \end{equation}
 while this transformation is allowed by the condition as Eq. (\ref{HDP_bound}).

 However, for $n$-qubit states satisfying the following two conditions, we construct a quantum operation $\cptp\in \HDP$ such that the equality sign in Eq. (\ref{HDP_bound}) is reached. The two conditions can be stated as follows.\\
 (C1) $\exists c\in \mathbb Z^+$, if $h(x,y)\neq c$, then $\rho_{xy}=0$.\\
 (C2) $\forall\rho_{xy},\rho_{x'y'}\neq 0$, where $x\neq x'$, $\mathrm{span}\{|x\rangle,|y\rangle\}$ is orthogonal to $\mathrm{span}\{|x'\rangle,|y'\rangle\}$.\\
 Condition (C1) means that the Hamming distance between the supports of each non-zero off-diagonal element equals to a constant. Condition (C2) means that the supports of any two  off-diagonal elements do not have overlap. Therefore, for any non-zero off-diagonal element $\rho_{xy}$, $y$ uniquely depends on $x$, and hence we write $y$ as $y(x)$.

 Assume the above two conditions are satisfied by the input state. Let $i,j$ be $n$-bit strings satisfying $h(i,j)=h(x,y(x))$. We construct $\cptp\in \HDP$ as $\cptp(\cdot)=K_0(\cdot)K_0^\dagger+\sum_x K_{x,1}(\cdot)K_{x,1}^\dagger+K_{x,2}(\cdot)K_{x,2}^\dagger$ with
 \begin{eqnarray}
 \label{eq:Kraus_merge}
     K_{x,1} &= &\sqrt{p_{i|x}}e^{-i\phi_x}|i\rangle\langle x|+\sqrt{p_{j|y(x)}}|j\rangle\langle y(x)|,\nonumber\\
      K_{x,2} &= &\sqrt{p_{j|x}}e^{-i\phi_x}|j\rangle\langle x|+\sqrt{p_{i|y(x)}}|i\rangle\langle y(x)|,\nonumber\\
      K_0&=&[\iden-\sum_x(K_{x,1}^\dagger K_{x,1}+K_{x,2}^\dagger K_{x,2})]^{\frac12},
 \end{eqnarray}
 where $x$ is the $n$-bit string such that $\rho_{xy}\neq 0$ and $x<y(x)$, and $\phi_x$ is the phase of $\rho_{xy}$.
It follows that
\begin{equation}
    \langle i|\cptp(\rho)|j\rangle=\sum_x (\sqrt{p_{i|x}p_{j|y(x)}}+\sqrt{p_{j|x}p_{i|y(x)}})|\rho_{xy}|,
\end{equation}
which means that the bound in Eq. (\ref{HDP_bound}) is reached by the channel $\cptp$.
Furthermore, if one aims to maximize $\langle i|\cptp(\rho)|j\rangle$ over all $\HDP$, we choose $p_{i|x}=p_{j|y(x)}$ and $p_{j|x}=p_{i|y(x)}=1-p_{i|x}$, and obtain
\begin{equation}
    \langle i|\cptp(\rho)|j\rangle_{\max}=\sum_x |\rho_{xy}|.
\end{equation}
It means that, if conditions (C1) and (C2) are satisfied by input state, then the off-diagonal elements of the same Hamming mode can be merged together by HDP.


	
\subsection{SELECTIVELY HAMMING DISTANCE PRESERVING OPERATIONS (SHP)}\label{subsection:SHP}
We give the definition of selectively Hamming distance preserving opertions (SHP) as follows.

\begin{definition}
	\label{SHP_D1}
 	Let $\shp:\mathcal{L}(\mathcal{H}^n)\mapsto\mathcal{L}(\mathcal{H}^n)$ be a CPTP map. Then $\shp$ belongs to selectively Hamming distance preserving operations (SHP) if it has a Kraus decomposition $\shp(\cdot)=\sum_l K_l(\cdot) K_l^\dagger$ such that
	\begin{equation}
		\label{SHP}
		\langle i|K_l |x\rangle\langle y| K_l^\dagger |j\rangle=0,\ \forall l,
	\end{equation}
 if $h(i,j)\neq h(x,y)$.
\end{definition}
By \textbf{Proposition \ref{prop:1}}, SHP is equivalent to the second set of free operations $\mathbb{F}_2$ defined in Eq.~(\ref{eq:F2}). Hence, we denote $\mathbb{F}_2$ as SHP here and after.

The following proposition gives the explicit form of the Kraus operators of a selectively Hamming distance preserving operation.
\begin{proposition}
	\label{SHP_STRUCTURE}
 If $\shp\in\SHP$, then $\shp$ has a Kraus decomposition $\shp(\cdot)=\sum_l K_l(\cdot)K_l^\dagger$ such that every Kraus operator can be written as
	\begin{equation}
		\label{FK}
		K_{l}=\sum_{x=0}^{2^n-1}c_{lx}|\pi_{l}^{n}(x)\rangle\langle x|,
	\end{equation}
	where  $\sum_l|c_{lx}|^2=1$,
 and $\pi_{l}^{n}:\NS\mapsto\NS$ is a one-to-one map such that 
	\begin{equation}
		\label{MAP}
		h(x,y)=h(\pi_{l}^{n}(x),\pi_{l}^{n}(y)),\forall x,y.
	\end{equation}
	We define $\{\pi_j^n\}$ as Hamming distance preserving functions and discuss them in \textbf{Appendix \ref{sec:appendix-a}}.
\end{proposition}

\begin{proof}
For $x=y$, Eq.~(\ref{SHP}) reduces to 
\begin{equation}
    \langle i|K_l|x\rangle\langle x|K_l^\dagger|j\rangle=\delta_{ij}|\langle i|K_l|x\rangle|^2.
\end{equation}
Similarly, we have
\begin{equation}
   \langle i|K_l|x\rangle\langle y|K_l^\dagger|i\rangle=\delta_{xy}|\langle i|K_l|x\rangle|^2.
\end{equation}
Hence, there is at most one non-zero element in each row and each column of $K_l$. The Kraus operators can then be written as 
	\begin{equation}
		K_l=\sum_{x=0}^{2^n-1}c_{lx}|f_l(x)\rangle\langle x|,
	\end{equation}
	where $f_l:\NS\mapsto\NS$ is a one-to-one map and the completion identity $\sum_lK_l^\dagger K_l=\iden$ demands $\sum_l|c_{lk}|^2=1$.

 In order to prove that $f_l(x)$ is a Hamming distance preserving function, we substitute the above form of $K_l$ in Eq.~(\ref{SHP}). 
It follows that the necessary condition for 
	\begin{equation}
		\langle i|f_l(x)\rangle \langle f_l(y)|j\rangle c_{lx}c_{ly}^*\neq0
	\end{equation}
	is that $f_l(x)=i,\ f_l(y)=j$, and $h(x,y)=h(i,j)$. In other words, for any two coefficients $c_{lx},c_{ly}\neq0$, it holds that
	\begin{equation}
		h(x,y)=h(f_l(x),f_l(y)),
	\end{equation}
or equivalently, $f_l(x)$ is a Hamming distance preserving function.
	\end{proof}

	For a resource theory, the states, which can be transformed to any state of the same dimension under free operations, are called golden states. To explore the existence of golden states is one of the most notable problems for a resource theory. Here we will prove that the uniform superposition states, denoted as
 \begin{equation}
     |+^{(n)}\rangle=\frac{1}{\sqrt{2^n}}\sum_{x=0}^{2^n-1}e^{-i\eta_x}|x\rangle,\ \eta_x\in[0,2\pi),
 \end{equation}
  are the golden states in the present resource theory.
\begin{proposition}
	\label{proposition_TS_SHP}
	An $n$-qubit uniform superposition state can be transformed to any $n$-qubit state by SHP.
\end{proposition}
\begin{proof}
 By definition, SHP is convex. Hence it is sufficient to prove that any $n$-qubit pure state can be obtained from $|+^{(n)}\rangle$ by SHP.
 The general form of an $n$-qubit pure state reads $|\psi^{(n)}\rangle=\sum_{x=0}^{2^n-1}e^{-i\phi_x}\sqrt{\psi_x}|x\rangle$, where $\phi_x\in[0,2\pi)$ and $\psi_x\geq0$ satisfying $\sum_x \psi_x=1$. 
 
 Here we design a selectively Hamming distance preserving operation $\mathcal{S}(\cdot)=\sum_{z=0}^{2^n-1}K_z(\cdot) K_z^\dagger$, which transforms $|+^{(n)}\rangle$ to $|\psi^{(n)}\rangle$. Each Kraus operator is in the following form
	\begin{equation}
		\label{STATE_TRANSFORMATION}
		K_z=\sum_{x=0}^{2^n-1}e^{-i(\phi_x-\eta_{q(z,x)})}\sqrt{\psi_x}|x\rangle\langle q(z,x)|,
	\end{equation}
	where  $q(z,x)$ is an $n$-bit string with the $m$th bit defined as
	\begin{equation}
		(q(z,x))_m=z_m\oplus x_m.
	\end{equation}
 Therefore, $\{|q(z,x)\rangle\}_z$ is a basis of $n$-qubit Hilbert space for any $x$, or equivalently, $\sum_{z=0}^{2^n-1}|q(z,x)\rangle\langle q(z,x)|=\iden$. It is directly checked that  $\sum_{z}K_z^\dagger K_z=\iden$, which ensures that $\mathcal{S}$ is a CPTP map. Besides, $h(x,x')=h(q(z,x),q(z,x'))$, which ensures that $\mathcal{S}\in\SHP$.
	
	Then we calculate 
	\begin{widetext}
		\begin{equation}
			\begin{split}
				K_z|+^{(n)}\rangle&=\dfrac{1}{\sqrt{2^n}}\sum_{x=0}^{2^n-1}e^{-i(\phi_x-\eta_{q(z,x)})}\sqrt{\psi_x}|x\rangle\langle q(z,x)|\sum_{y=0}^{2^n-1}e^{-i\eta_y}|y\rangle\\
				&=\dfrac{1}{\sqrt{2^n}}\sum_{x=0}^{2^n-1}e^{-i\phi_x}\sqrt{\psi_x}|x\rangle\\
				&=\dfrac{1}{\sqrt{2^n}}|\psi^{(n)}\rangle.			
			\end{split}
		\end{equation}
	\end{widetext}
 Consequently, $\mathcal{S}(|+^{(n)}\rangle\langle +^{(n)}|)=|\psi^{(n)}\rangle\langle\psi^{(n)}|$.
\end{proof}

The set of golden states contains $n$-qubit product states. For such states, the resource of entanglement does not provide any advantage in dephasing estimation. In fact, $\PDqfi(|+^{(n)}\rangle\langle +^{(n)}|)= n\PDqfi(|+^{(1)}\rangle\langle+^{(1)}|)$. This is consistent with the results in \cite{Kolodynski_2013, PhysRevLett.118.100502, Pirandola2018}, where it is found that dephasing estimation cannot beat the standard quantum limit.

It is interesting to notice that the set of golden states also contains maximally entangled states, e.g., 
the two-qubit state $\frac12(|00\rangle+|01\rangle+|10\rangle-|11\rangle)$. For such states, the local coherence of each qubit vanishes but the entanglement reaches maximum.
	
	
	Next, we introduce the concept of Hamming distance preserving unitary operations, which are free unitary operations in the resource theory of dephasing estimation.
	\begin{definition}
		Let $\mathcal{U}(\cdot)=U_j(\cdot) U_j^{\dagger}$ be a unitary operation, where $U_j$ is a unitary operator. Then $\mathcal{U}$ is said to be a Hamming distance preserving unitary operation if and only if
		\begin{equation}
			\label{FU}
			\pd(U_j\rho {U_j}^{\dagger})=U_j\pd(\rho){U_j}^{\dagger},\ \forall \rho.
		\end{equation}
	\end{definition}
 Because a Hamming distance preserving unitary operation belongs to $\SHP$, from \textbf{Proposition \ref{SHP_STRUCTURE}}, we arrive at following structure of these unitary operators.
	\begin{corollary}   
	\label{p1}
		Any $n$-qubit unitary operator satisfies Eq.~(\ref{FU}) if and only if it can be written as the form
	\begin{equation}
		U_j=\sum_{x=0}^{2^n-1}e^{-i\omega_x}|\pi_j^n(x)\rangle \langle x| \quad \forall j, \rho,
	\end{equation}
	where $\pi_j^n(\cdot)$ is a Hamming distance preserving function.
	\end{corollary}
	It is worth noticing that the inverse of $U_j$ also belongs to free unitary operations for dephasing estimation.  
	It follows that the probe state preserves the ability of dephasing estimation under Hamming distance preserving unitary operations.

\subsection{THE RELATION BETWEEN SHP AND HDP}\label{subsection:HDP_SHP_relation}
    In this subsection, we investigate the relation of SHP and HDP for single-qubit and multiqubit states.
    We first review the definition of the Choi-Jamiołkowski matrix, which plays a role in exploring the relation of SHP and HDP. The Choi-Jamiołkowski matrix of the $d$-dimension operation $\mathcal{E}$ is defined as \cite{PhysRevA.103.022403, PhysRevA.80.022339}
	\begin{equation}
		J_\mathcal{E}=\begin{bmatrix}
			\mathcal{E}(|0\rangle\langle0|)&\cdots&\mathcal{E}(|0\rangle\langle j|)&\cdots&	\mathcal{E}(|0\rangle\langle d|)\\
			\vdots&\ddots&\vdots&\ddots&\vdots\\
			\mathcal{E}(|i\rangle\langle0|)&\cdots&\mathcal{E}(|i\rangle\langle j|)&\cdots&	\mathcal{E}(|i\rangle\langle d|)\\
			\vdots&\ddots&\vdots&\ddots&\vdots\\
			\mathcal{E}(|d\rangle\langle0|)&\cdots&\mathcal{E}(|d\rangle\langle j|)&\cdots&	\mathcal{E}(|d\rangle\langle d|)\\
		\end{bmatrix}.
	\end{equation}
\begin{proposition}
	\label{HDP_SIO_relation_1}
	HDP is equivalent to SHP in the single-qubit case.
\end{proposition}
\begin{proof}
	Let $\mathcal{D}(\cdot)=\sum_{a}D_a(\cdot)D_a^{\dagger}$ be a single-qubit CPTP map which belongs to HDP. When the operation $\mathcal{D}$ belongs to $\HDP$, the corresponding Choi-Jamiołkowski matrix is written as 
	\begin{equation}
		\label{D_Choi}
		J_\mathcal{D}=\begin{bmatrix}
		p_{0|0}&0&0&\gamma_0\\
		0&p_{1|0}&\gamma_1&0\\
		0&\gamma_1^*&p_{0|1}&0\\
		\gamma_0^*&0&0&p_{1|1}
		\end{bmatrix},
	\end{equation}
	where $\gamma_0=\sum_a\langle0|D_a|0\rangle\langle1|D_a^\dagger|1\rangle$, $\gamma_1=\sum_a\langle1|D_a|0\rangle\langle1|D_a^\dagger|0\rangle$, and 
 $p_{x|y}=\sum_a\langle x|D_a|y\rangle\langle y|D_a^\dagger|x\rangle$ with $x,y=0,1$.

 Now consider an operation $\mathcal{S}\in\SHP$, which is related to $\mathcal{D}$ as
 $\mathcal{S}(\cdot)=\sum_{a}S_a(\cdot)S_a^\dagger+\sum_{a}T_a(\cdot)T_a^\dagger$ with
	\begin{equation}
		\label{SHP=HDP}
		\begin{split}
			S_a&=\begin{bmatrix}
				\langle0|D_a|0\rangle&0\\
				0&\langle1|D_a|1\rangle
			\end{bmatrix},\\
			T_a&=\begin{bmatrix}
				0&\langle0|D_a|1\rangle\\
				\langle1|D_a|0\rangle&0
			\end{bmatrix}.
		\end{split}
	\end{equation}
 It is directly checked that the Choi-Jamiołkowski matrix of $\mathcal{S}$ is also Eq.~(\ref{D_Choi}), and therefore, $\mathcal{S}=\mathcal{D}$.
\end{proof}
Furthermore, in \textbf{Appendix \ref{sec:appendix-a_0}} we prove that, for single-qubit states, $\HDP$ is a strict subset of operations which do not increase the QFI of dephasing estimation.
\begin{proposition}
	In the multiqubit case, SHP $\subsetneq$ HDP.
\end{proposition}
	SHP is a subset of HDP by their definitions. Then we give a channel which belongs to HDP but not to SHP. The channel can be written as $\mathcal{W}(\cdot)=\sum_{i=0}^{4}W_i(\cdot)W_i^{\dagger}$, where
	\begin{equation}
		\begin{split}
			\label{DIFF_SHP}
			W_0&=\dfrac{1}{2}|0\rangle\langle 1|+\dfrac{1}{2\sqrt{2}}(|1\rangle\langle 0|+|2\rangle\langle 0|),\\
			W_1&=\dfrac{1}{2}|1\rangle\langle 1|+\dfrac{1}{2\sqrt{2}}(|0\rangle\langle 0|+|3\rangle\langle 0|),\\
			W_2&=\dfrac{1}{2}|2\rangle\langle 1|+\dfrac{1}{2\sqrt{2}}(|0\rangle\langle 0|-|3\rangle\langle 0|),\\
			W_3&=\dfrac{1}{2}|3\rangle\langle 1|+\dfrac{1}{2\sqrt{2}}(|1\rangle\langle 0|-|2\rangle\langle 0|),\\
			W_4&=\iden-|0\rangle\langle 0|-|1\rangle\langle 1|.
		\end{split}		
	\end{equation}
	By the definition of HDP Eq.~(\ref{D_HDP}), this channel $\mathcal{W}$ is a Hamming distance preserving operation. The following lemma is used to prove that $\mathcal{W}$ does not belong to SHP.
	\begin{lemma}[Theorem 8.2 of \cite{nielsen2002quantum}]
		\label{U_F_Operator}
	Consider two given quantum channels $\mathcal{E}(\cdot)=\sum_{i=1}^mE_i(\cdot) E_i^\dagger$ and $\mathcal{F}(\cdot)=\sum_{j=1}^nF_j(\cdot) F_j^\dagger$. 
 Then $\mathcal{E}=\mathcal{F}$ if and only if there is a $m\times n$ linear isometry $u$ such that $E_i=\sum_ju_{ij}F_j$.
	\end{lemma}
We find that the linear combination of $\left\lbrace W_i \right\rbrace$ cannot be expressed as Eq.~(\ref{FK}). Based on \textbf{Lemma\ \ref{U_F_Operator}.}, $\mathcal{W}$ is not a selective Hamming distance preserving operation.

\section{THE COMPARISON BETWEEN THE RESOURCE THEORIES OF DEPHASING ESTIMATION AND QUANTUM COHERENCE}
\label{COMPARISON}
In this section, we compare free operations between the resource theories of dephasing estimation and quantum coherence. Our results show that coherence cannot be regarded as the quantum resource which underlies the precision of dephasing estimation.
\subsection{A BRIEF REVIEW OF THE RESOURCE THEORY OF QUANTUM COHERENCE}
In the resource theory of quantum coherence \cite{RevModPhys.89.041003}, a computational basis $\{|i\rangle\}_{i=0}^{2^n-1}$ is prefixed. The completely dephasing (CD) channel is defined as
	\begin{equation}
	\cd(\cdot)=\sum_{i=0}^{2^n-1}|i\rangle\langle i|\cdot|i\rangle\langle i|.
	\end{equation}
 Then a state $\rho$ is free if and only if $\Delta(\rho)=\rho$. These free states are called incoherent states.

 In the following, we briefly review two sets of free operations. The first set is called dephasing-covariant incoherent operations (DIO) \cite{PhysRevA.94.052336, PhysRevLett.117.030401}, defined as quantum operations which commute with the CD channel, i.e.,
 \begin{equation}
 \mathcal{E}\in\DIO \Leftrightarrow \cd\circ\mathcal{E}=\mathcal{E}\circ\cd.
\end{equation}
The second set is called strictly incoherently operations (SIO) \cite{PhysRevLett.116.120404, PhysRevA.94.052336}, defined as
 \begin{equation}
 \mathcal{E}\in\SIO \Leftrightarrow \cd\circ\mathcal{E}_j=\mathcal{E}_j\circ\cd,\ \forall j,
\end{equation}
where $\mathcal{E}=\sum_j\mathcal{E}_j,\ \mathcal{E}_j(\cdot)=E_j(\cdot) E_j^\dagger$, and $E_j$ are Kraus operators satisfying $\sum_jE_j^\dagger E_j=\iden$.
It is worth noting that, the form of every Kraus operator $E_j$ can be written as
	\begin{equation}
			\label{SIO}
			E_j=\sum_{x=0}^{2^n-1}c_{jx}|e_j(x)\rangle\langle x|,
	\end{equation}
	where $e_j:\NS\mapsto\NS$ is one-to-one and $\sum_j|c_{jx}|^2=1$. 
 
It is worth noticing that, the resource theory of coherence is symmetric under permutations of states in the computational basis. 
Direct calculation shows that, the widely studied sets of free operations, including maximally incoherent operations (MIO), incoherent operations (IO), DIO, SIO, and physically incoherent operations (PIO) (see Ref. \cite{PhysRevLett.117.030401} for their definitions and comparisons), are closed under permutations of states in the computational basis.
Therefore, the two states $\sum_i \psi_i|i\rangle$ and $\sum_i \psi_i |\pi(i)\rangle$, where $\pi$ is a permutation, should contain the same amount of coherence.



\subsection{SINGLE QUBIT CASE}
We first explore the relation between the sets of free operations in the single qubit case.
\begin{proposition}
	\label{eq_SIO_SHP}
	In the single-qubit case, SHP, HDP, SIO, and DIO are equivalent to each other.
\end{proposition}
\begin{proof}
It has been proved in Ref. \cite{PhysRevA.94.052336} that DIO$=$SIO in the single-qubit case. In \textbf{Proposition \ref{HDP_SIO_relation_1}}, we proved that $\SHP=\HDP$. Therefore, here we only need to prove the equivalence between $\SIO$ and $\SHP$.

According to \cite{PhysRevLett.119.140402}, the general form of a single-qubit strictly incoherent operation is $\mathcal{E}^{\mathrm{SIO}}(\cdot)=\sum_{l=1}^4K_l(\cdot) K_l^\dagger$, where
    \begin{equation}
    \begin{split}
        K_1&=\begin{bmatrix}a_1&0\\0&b_1\end{bmatrix},K_2=\begin{bmatrix}0&b_2\\a_2&0\end{bmatrix},\\
        K_3&=\begin{bmatrix}a_3&0\\0&0\end{bmatrix},K_4=\begin{bmatrix}0&0\\a_4&0\end{bmatrix},
    \end{split}
    \end{equation}
	where $a_i$ is real for $i=1,2,3,4$ and $\sum_{i=1}^{4}a^2_i=\sum_{j=1}^2\left | b_j\right|^2=1$. These Kraus operators satisfy Eq.~(\ref{SHP}), so SIO $\subset$ SHP for single-qubit states. Moreover, \textbf{Proposition \ref{SHP_STRUCTURE}} implies SHP is a subset of SIO because Hamming distance preserving functions are one-to-one. Hence SHP$=$SIO. This completes the proof.
\end{proof}
    In fact, the HDP cone Eq.~$(\ref{SIO_region})$ is also the SIO cone \cite{PhysRevLett.119.140402, Shi2017}, which reveals the equivalence of SIO and HDP for single-qubit states. 

\subsection{MULTIQUBIT CASE}
In this subsection, we explore the comparison between SHP, HDP, SIO, and DIO in the multiqubit case. 
The hierarchy of these sets of operations is shown in Fig. \ref{f_2}.
It is inferred by the definitions that $\SIO\subseteq\DIO$, $\HDP\subseteq\DIO$ and $\SHP\subseteq\SIO\cap\HDP$.

	\begin{figure}
		\centering
		\includegraphics[width=0.45\textwidth]{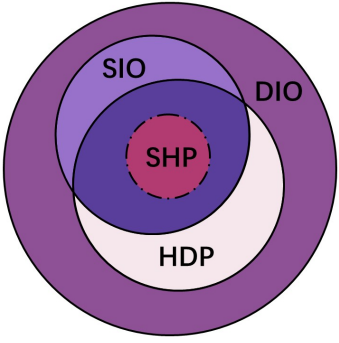}
		\caption{The comparison between SHP, HDP, SIO and DIO in the multiqubit case. The boundary between $\SHP$ and the intersection of $\SIO$ and $\HDP$ is marked with the chain line. This boundary is unclear in the $2$-qubit case but clear in the $m$-qubit case, where $m\geq3$.}
		\label{f_2}
	\end{figure}

In the following, we will first show that neither $\SIO$ nor $\HDP$ contains the other quantum operations belonging to $\SIO\setminus\HDP$ and $\HDP\setminus\SIO$.
 Consider an $n$-qubit ($n\geq2$) unitary operator
\begin{equation}
	U=\iden-|2\rangle\langle2|+|2\rangle\langle3|+|3\rangle\langle2|-|3\rangle\langle3|.
\end{equation}
The corresponding unitary operation belongs to $\SIO$ because it satisfies Eq.~(\ref{SIO}). However, because $U|0\rangle\langle2|U^\dagger=|0\rangle\langle3|$ and $h(0,2)=1\neq 2=h(0,3)$, this unitary operation does not belong to $\HDP$.

Next consider the $n$-qubit ($n\geq2$) operation $\mathcal{W}$ in Eq.~(\ref{DIFF_SHP}). It belongs to $\HDP$. It can be checked that any linear combination of the five Kraus operators of $\mathcal{W}$ is not of the form as Eq.~(\ref{SIO}). Based on Lemma \ref{U_F_Operator}, $\mathcal{W}\in\HDP\setminus\SIO$.

Now we will show that $\SHP$ is a strict subset of $\SIO\cap\HDP$ in the $m$-qubit case, where $m\geq3$. This is achieved by constructing the following channel $\mathcal{R}(\cdot)=\sum_{i=0}^{4}R_i(\cdot) R_i^{\dagger}$, where
	\begin{equation}
		\begin{split}
			R_0&=\dfrac{1}{\sqrt{2}}|0\rangle \langle0|+\dfrac{1}{2}|2\rangle \langle1|+\dfrac{1}{2}|3\rangle \langle6|,\\
			R_1&=\dfrac{1}{\sqrt{2}}|6\rangle \langle0|+\dfrac{1}{2}|2\rangle \langle1|-\dfrac{1}{2}|3\rangle \langle6|,\\
			R_2&=\dfrac{1}{2}|1\rangle \langle1|+\dfrac{1}{2}|6\rangle \langle6|+\dfrac{1}{\sqrt{2}}|7\rangle \langle7|,\\
			R_3&=\dfrac{1}{2}|4\rangle \langle1|-\dfrac{1}{2}|3\rangle \langle6|+\dfrac{1}{\sqrt{2}}|7\rangle \langle7|,\\
			R_4{\Large }&=\iden-|0\rangle \langle0|-|1\rangle \langle1|-|6\rangle \langle6|-|7\rangle \langle7|.\\
		\end{split}
	\end{equation}
The channel $\mathcal{R}$ belongs to the intersection of SIO and HDP but not to SHP by \textbf{lemma \ref{U_F_Operator}}, that is, $\mathcal{R}\in \SIO\cap\HDP\setminus\SHP$ for $m$-qubit states. 
However, it is unclear whether $\SHP=\SIO\cap\HDP$ in the $2$-qubit case.


Finally, we show that $\SIO\cup\HDP$ is a strictly subset of $\DIO$. As proved in Ref. \cite{2016arXiv160406524B}, the quantum operation $\mathcal{N}(\cdot)=\sum_{i=0}^{4}N_i(\cdot)N_i^\dagger$ with
	\begin{equation}
		\begin{split}
			N_0&=\dfrac{1}{2}|0\rangle\langle1|+\dfrac{1}{2\sqrt{3}}|1\rangle\langle0|-\dfrac{1}{2\sqrt{3}}|2\rangle\langle0|+\dfrac{1}{2\sqrt{3}}|3\rangle\langle0|,\\
			N_1&=\dfrac{1}{2\sqrt{3}}|0\rangle\langle0|+\dfrac{1}{\sqrt{2}}|0\rangle\langle2|+\dfrac{1}{\sqrt{6}}|0\rangle\langle3|+\dfrac{1}{2}|1\rangle\langle1|\\
			&+\dfrac{1}{2\sqrt{3}}|2\rangle\langle0|+\dfrac{1}{2\sqrt{3}}|3\rangle\langle0|,\\
			N_2&=\dfrac{1}{2\sqrt{3}}|0\rangle\langle0|-\dfrac{1}{\sqrt{2}}|0\rangle\langle2|+\dfrac{1}{\sqrt{6}}|0\rangle\langle3|+\dfrac{1}{2\sqrt{3}}|1\rangle\langle0|\\
			&+\dfrac{1}{2}|2\rangle\langle1|-\dfrac{1}{2\sqrt{3}}|3\rangle\langle0|,\\
			N_3&=\dfrac{1}{2\sqrt{3}}|0\rangle\langle0|-\dfrac{\sqrt{6}}{3}|0\rangle\langle3|-\dfrac{1}{2\sqrt{3}}|1\rangle\langle0|-\dfrac{1}{2\sqrt{3}}|2\rangle\langle0|\\
			&+\dfrac{1}{2}|3\rangle\langle1|,\\
			N_4&=\iden-(|0\rangle\langle0|+|1\rangle\langle1|+|2\rangle\langle2|+|3\rangle\langle3|),
		\end{split}
	\end{equation}
 belongs to $\DIO\setminus\SIO$. Because $\mathcal{N}(|1\rangle\langle2|)=\dfrac{1}{2\sqrt{2}}(|1\rangle\langle0|-|2\rangle\langle0|)$, and $h(1,2)\neq h(0,1)$ (or $h(1,2)\neq h(0,2)$), $\mathcal{N}$ is not in $\HDP$. Hence, the set $\DIO\setminus(\SIO\cup\HDP)$ is not empty.

Different from the sets of free operations in the resource theory of coherence, neither SHP nor HDP is closed under permutations of states in the computational basis. Therefore, off-diagonal elements of a density matrix are not equivalent in the present resource theory if they belong to different Hamming modes. For example, the states $|\psi_1\rangle=(|0\rangle+|1\rangle)/\sqrt{2}$ and $|\psi_2\rangle=(|0\rangle+|3\rangle)/\sqrt{2}$ contain the same amount of coherence, but as $|0\rangle\langle1|$ (or $|1\rangle\langle0|$) and $|0\rangle\langle3|$ (or $|3\rangle\langle0|$) belong to different Hamming modes, the powers of the two probe states in dephasing estimation are different.
In fact, we have
$\PDqfi(|\psi_1\rangle\langle\psi_1|)=e^{-2\theta}/(1-e^{-2\theta})$, but $\PDqfi(|\psi_2\rangle\langle\psi_2|)=4e^{-4\theta}/(1-e^{-4\theta})$. Therefore, despite some similarities between the two resource theories, coherence cannot be regarded as the quantum resource underlying the precision of dephasing estimation.

\section{THE COMPARISON BETWEEN THE RESOURCE THEORIES OF DEPHASING ESTIMATION AND $U(1)$ ASYMMETRY}
\label{COMPATIBILITY}
The QFI of phase estimation is a resource measure of $U(1)$ asymmetry \cite{PhysRevA.93.022122, PhysRevA.90.062110,PhysRevLett.129.190502}. 
It has been shown that there is a probe incompatibility between phase estimation and dephasing estimation, namely, no optimal probe states exist for detecting phase and dephasing parameters simultaneously \cite{PhysRevA.94.052108, PhysRevX.12.011039}.
This motivates us to compare the resource theories of $U(1)$ asymmetry and dephasing estimation.

\subsection{A BRIEF REVIEW OF QUANTUM FISHER INFORMATION OF PHASE ESTIMATION}
	In this subsection, we briefly review the QFI of phase estimation. Given Hamiltonian $H$ and the probe state $\rho$, $\rho(t)\equiv e^{-iHt}\rho e^{iHt}$ represents the state of system at time $t$. The task of phase estimation is to estimate the parameter $t$ by quantum measurement.
 Let $\rho=\sum_{i}p_i|\psi_i\rangle \langle\psi_i|$ be the spectral decomposition of the probe state $\rho$.
 The QFI of phase estimation can be calculated as \cite{BRAUNSTEIN1996135,PhysRevLett.72.3439, Toth_2014}
 \begin{equation}
	\PEqfi(\rho)=2\sum_{i,k}\dfrac{(p_i-p_k)^2}{p_i+p_k}\left|  \langle \psi_i|H|\psi_k\rangle\right|^2.
\end{equation}
Furthermore, when the probe is in a pure state $|\psi\rangle \langle\psi|$, the QFI reduces to
	\begin{equation}
		\PEqfi(|\psi\rangle \langle\psi|)=4 \Delta H^2,
	\end{equation}
	where $\Delta H^2:=\langle\psi | H^2|\psi\rangle-\langle\psi| H|\psi\rangle^2$ is the variance of Hamiltonian $H$.
 
In this section, we consider the setting of $n$ qubits with no interaction. For each qubit, the ground state energy is zero and the excited state energy is $\epsilon$. The form of $n$-qubit Hamiltonian is given by
	\begin{equation}
		H=\sum_{i=0}^{2^n-1}N^i_1\epsilon|i\rangle \langle i|,
	\end{equation}
	where $N^i_1$ represents the number of $1$'s in the binary string $i$.

\subsection{SINGLE QUBIT CASE}
	For $n=1$, the Hamiltonian reads $H=\epsilon|1\rangle \langle 1|$. For a single-qubit probe state $\sigma$, whose Bloch vector is $(r\cos{\phi},r\sin{\phi},z)$, its spectral decomposition is
	\begin{equation}
		\sigma=\frac{1}{2}[(1+\sqrt{z^2+r^2})|+r\rangle \langle +r|+(1-\sqrt{z^2+r^2})|-r\rangle \langle -r|],
	\end{equation}
	where eigenstates are given by 
	\begin{eqnarray}
		|+r\rangle&=&\dfrac{(z+\sqrt{z^2+r^2})|0\rangle+re^{i\phi}|1\rangle}{\sqrt{r^2+(z+\sqrt{z^2+r^2})^2}},\nonumber\\
		|-r\rangle&=&\dfrac{-r|0\rangle+(z+\sqrt{z^2+r^2})e^{i\phi}|1\rangle}{\sqrt{r^2+(z+\sqrt{z^2+r^2})^2}}.
	\end{eqnarray}
	Then the phase estimation QFI of $\sigma$ is calculated as
\begin{equation}\label{eq:peqfi_quibt}
	\PEqfi(\sigma)=r^2\epsilon^2.
\end{equation}
	Therefore, in the single-qubit case, $\PEqfi$ cannot be increased under HDP because Eq.~(\ref{SIO_region}) shows that $r$ is non-increasing under HDP.

 Nevertheless, by comparing Eqs. (\ref{eq:peqfi_quibt}) and (\ref{eq:pdqfi_qubit}), we find that $\PEqfi$ and $\PDqfi$ give different ordering of states even in the single-qubit case. It means that even though both $\PEqfi$ and $\PDqfi$ reaches maximum for uniform superposition states, there is a possibility that for two single-qubit states $\sigma_1$ and $\sigma_2$, $\PEqfi(\sigma_1)>\PEqfi(\sigma_2)$ but $\PDqfi(\sigma_1)<\PDqfi(\sigma_2)$.
 
\subsection{MULTIQUBIT CASE}
In the multiqubit case, the golden states in the resource theory of dephasing estimation are uniform superposition states $|+^{(n)}\rangle$, while $\PEqfi$ reaches its maximum for GHZ states defined as $|GHZ\rangle=(|0\rangle+e^{-i\eta_0}|2^n-1\rangle)/\sqrt{2}$ \cite{PhysRevLett.96.010401}. This discordance in the maximally resourceful states is a possible origin of the probe incompatibility, which stands for the absence of an optimal probe state for detecting phase and dephasing parameters simultaneously \cite{PhysRevX.12.011039, PhysRevA.94.052108}.
It is then natural to ask the following question: in what circumstances, can one increase $\PEqfi$ of a probe state, and meanwhile preserve its ability of dephasing estimation?

An illustrative example is as follows. Consider a two-qubit pure state $|\psi_1\rangle=\frac{1}{\sqrt2}(|1\rangle+|2\rangle)$. Because it is an eigenstate of the Hamiltonian, we have $\PEqfi(|\psi_1\rangle\langle \psi_1|)=0$. However, by applying a Hamming distance preserving unitary 
	\begin{equation}
		V=\begin{bmatrix}
			0&1&0&0\\
			1&0&0&0\\
			0&0&0&1\\
			0&0&1&0\\
		\end{bmatrix},
	\end{equation}
 $|\psi_1\rangle$ is converted to $|\psi_2\rangle=V|\psi_1\rangle=(|0\rangle+|3\rangle)/\sqrt{2}$, which is one of the two-qubit optimal probe states for phase estimation. It means that a free unitary operation in the resource theory of dephasing estimation can transform a free state in the resource theory of $U(1)$ asymmetry to a maximally resourceful state in the latter resource theory. This shows the incompatibility of the two resource theories.

This example is generalized to general two-qubit pure states as follows.
\begin{proposition}
Let $|\psi^{(2)}\rangle=\sum_{x=0}^{3}e^{-i\phi_x}\sqrt{\psi_x}|x\rangle$ be a two-qubit pure state.	Hamming distance preserving unitary operations can increase the QFI of $|\psi^{(2)}\rangle$ in phase estimation, if and only if $g(\psi_1,\psi_2)> g(\psi_0,\psi_3)$, where $g(x,y)\equiv x+y-(x-y)^2$.
\end{proposition}
\begin{proof}
		The QFI of the state $|\psi^{(2)}\rangle$ in phase estimation is calculated as
			\begin{equation}
				\label{2_qubit_pure_PEQFI}
					\PEqfi(|\psi^{(2)}\rangle \langle \psi^{(2)}|)=g(\psi_0,\psi_3)\epsilon^2.
			\end{equation}
After the action of the unitary $V$, the QFI of phase estimation becomes 
\begin{equation}
		\label{HDPUO_2_qubit_pure_PEQFI}
		\PEqfi(V|\psi^{(2)}\rangle \langle \psi^{(2)}|V^\dagger)=g(\psi_1,\psi_2)\epsilon^2,
\end{equation}
	which is larger than Eq.~(\ref{2_qubit_pure_PEQFI}) if $g(\psi_1,\psi_2)> g(\psi_0,\psi_3)$. 
 
 Moreover, when other Hamming preserving distance unitary operations are performed on the state $|\psi^{(2)}\rangle$, the resulted $\PEqfi$ is either in the form of Eq.~(\ref{2_qubit_pure_PEQFI}) or Eq.~(\ref{HDPUO_2_qubit_pure_PEQFI}). Therefore, the QFI of $|\psi^{(2)}\rangle$ in phase estimation can be increased only if $g(\psi_1,\psi_2)> g(\psi_0,\psi_3)$.
\end{proof}

Interestingly, some nonunitary operations in SHP also have the ability in improving $\PEqfi$ while preserving $\PDqfi$. Consider an $n$-qubit mixed state in the following form
\begin{equation}
    \rho=\sum_{x:x<\Bar{x}} p_x |\psi(x)\rangle\langle \psi(x)|,
\end{equation}
where $|\psi(x)\rangle=\cos\zeta|x\rangle+\sin\zeta e^{i\phi_x}|\Bar{x}\rangle$, $\Bar{x}$ stands for the bitwise bit flip of $x$, and $\zeta$ is independent of $x$. It satisfies conditions (C1) and (C2) for coherence merge in the same Hamming mode. Therefore, we apply the merging operation defined in Eq. (\ref{eq:Kraus_merge}) with $i=0$, $j=2^n-1$, $y(x)=\Bar{x}$, $p_{i|x}=p_{j|\Bar{x}}=1$ and $p_{j|x}=p_{i|\Bar{x}}=0$. This results in a pure state of the form
\begin{equation}
    \rho'=|\psi(0)\rangle\langle \psi(0)|.
\end{equation}
By the convexity of QFI, we have $\PEqfi(\rho')>\sum_x p_x\PEqfi(|\psi(x)\rangle\langle \psi(x)|)\geq \PEqfi(\rho)$. This means that $\PEqfi$ is increased. Meanwhile, $\PDqfi(\rho)=\PDqfi(\rho')$. The reason is as follows. On the one hand, from the monotonicity of $\PDqfi$ under SHP, we have $\PDqfi(\rho)\geq\PDqfi(\rho')$. On the other hand, from the convexity of QFI, we have $\PDqfi(\rho')=\sum_x p_x\PDqfi(|\psi(x)\rangle\langle \psi(x)|)\geq\PDqfi(\rho)$.

To sum up, one can use coherence merge in the same Hamming mode to increase the ability of a probe state in phase estimation, and meanwhile, preserve its ability of dephasing estimation.


 
\section{CONCLUSION}
\label{CONCLUSION}
In this work, we study the resource theory of dephasing estimation in multiqubit systems. Due to the monotonicity of QFI, we define two sets of free operations in the resource theory of dephasing estimation, HDP and SHP. Under these free operations, the problem of finding optimal probes can be converted to the problem of state transformations. This resource-theoretic approach reduces the complexity to calculate the QFI. In the present resource theory, the uniform superposition states are golden states. This implies that there exists the connection between our resource theory and the resource theory of quantum coherence. Therefore, we investigate the relation between free operations of these two resource theories. In the single-qubit case, SIO and DIO, which are free operations of the resource theory of quantum coherence, are equivalent to SHP (or HDP). Furthermore, for multiqubit states, the relation of these operations is depicted in FIG.$\ref{f_2}$. Finally, we also compare the resource theories of $U(1)$ asymmetry and dephasing estimation. By employing SHP, it is possible to enhance the performance of a probe state in phase estimation, while preserving its ability of dephasing estimation, which implies the incompatibility of these two resource theories.
\begin{acknowledgments}
This work was Supported by National Natural Science Foundation of China under Grant No. 11774205.
\end{acknowledgments}

\appendix
\section{SINGLE-QUBIT HDP IS A STRICT SUBSET OF SINGLE-QUBIT OPERATIONS WHICH DOES NOT INCREASE THE QFI OF DEPHASING ESTIMATION}
\label{sec:appendix-a_0}
For a single-qubit state $\sigma=(r\cos\phi,r\sin\phi,z)$, the QFI of dephasing estimation is independent on the phase $\phi$ due to the rotational symmetry. Therefore, we only consider the case $\phi=0$ in the following discussion. For a single-qubit state $\sigma_0=(r,0,z)$, the QFI of dephasing estimation is
\begin{equation}\label{eq:pdqfi_qubit}
\PDqfi(\sigma_0)=\dfrac{r^2e^{-2\theta}(1-z^2)}{1-z^2-r^2e^{-2\theta}}.
\end{equation}
In the following, we will show a channel $\mathcal{Z}\notin\HDP$, which cannot increase the QFI of phase estimation. The Choi matrix of the channel $\mathcal{Z}$ is given by
\begin{equation}
    J_\mathcal{Z}=\begin{bmatrix}
		\frac{1}{2}&0&\frac{1}{4}&\frac{1}{4}\\
		0&\frac{1}{2}&\frac{1}{4}&-\frac{1}{4}\\
		\frac{1}{4}&\frac{1}{4}&\frac{1}{2}&0\\
		\frac{1}{4}&-\frac{1}{4}&0&\frac{1}{2}
		\end{bmatrix}.
\end{equation}
Its Kraus operators are given by
\begin{equation}
\begin{split}
Z_0&=\dfrac{1}{2}\begin{bmatrix}1&0\\0&1\end{bmatrix}, Z_1=\dfrac{1}{2}\begin{bmatrix}0&1\\1&0\end{bmatrix},\\Z_2&=\dfrac{1}{2}\begin{bmatrix}1&1\\0&0\end{bmatrix}, Z_3=\dfrac{1}{2}\begin{bmatrix}0&0\\1&-1\end{bmatrix}.\end{split}
\end{equation}
Then the Bloch vector of $\mathcal{Z}(\sigma_0)$ can be written as $(\frac{1}{2}r,0,\frac{1}{2}r)$. The corresponding QFI of dephasing estimation is
\begin{equation}
    \PDqfi(\mathcal{Z}(\sigma_0))=\dfrac{\frac{1}{4}r^2e^{-2\theta}(1-\frac{1}{4}r^2)}{1-\frac{1}{4}r^2-\frac{1}{4}r^2e^{-2\theta}}
\end{equation}
When $r=0$, we find $\PDqfi(\sigma_0)=\PDqfi(\mathcal{Z}(\sigma_0))=0$. Then we discuss the case $r\neq0$. The inequality, $\PDqfi(\sigma_0)\geq \PDqfi(\mathcal{Z}(\sigma_0))$, is equivalent to
\begin{equation}
    \dfrac{3}{r^2e^{-2\theta}}-\dfrac{4}{4-r^2}+\dfrac{1}{1-z^2}\geq 0, \forall r,z.
\end{equation}
 When $z^2=0$ and $r^2=1$, the left side of the inequality reaches its minimum $3e^{2\theta}-1/3$, which is always greater than $0$. Therefore, the channel $\mathcal{Z}$ cannot increase the QFI of dephasing estimation and does not belong to $\HDP$.

\section{HAMMING DISTANCE PRESERVING FUNCTION}
\label{sec:appendix-a}
In this appendix, we first give the definition of a Hamming distance preserving function and then derive the explicit form of these functions.
\begin{definition}
Let $\NS$ be the set of $n$-bit strings.
	The one-to-one function $\pi^{n}_i:\NS\mapsto\NS$ is said to be a Hamming distance preserving function if 
	\begin{equation}
		h(x,y)=h(\pi^n_i(x),\pi^n_i(y)),\forall x,y\in\NS.
	\end{equation}
The set of $n$-bit Hamming distance preserving functions is denoted as $M^n$.
\end{definition}

Next we study the explicit form of Hamming distance preserving functions. We start with the simplest case where the input/output string consists of a single bit. 
Obviously, $M^1=\{\pi^1_0,\pi_1^1\}$ with
\begin{eqnarray}
	\pi^1_0: &0\mapsto0,1\mapsto1\nonumber,\\
	\pi_1^1: &0\mapsto1,1\mapsto0.
\end{eqnarray}

In the $2$-bit case, there are eight elements in $M^2$. Precisely, we denote by $x_m$ the $m$-th bit of string $x$, and define two sets of $2$-bit functions,
$R^2=\{r_0^2,r_1^2\}$ with $r_0^2(x_1x_0)=x_1x_0$, $r_1^2(x_1x_0)=x_0x_1$, and $Q^2=\{q_z\}_{z=0}^3$ with $(q_z(x))_l=z_l\oplus x_l,\ l=0,1$. Then the set of two-bit Hamming distance preserving functions can be written as
\begin{equation}
    M^2=Q^2\circ R^2.
\end{equation}
Here by $\circ$ we mean $A\circ B=\{m|m=ab,a\in A, b\in B\}$. This $2$-bit case can be generalized to the $n$-bit case.


\begin{proposition}
   The set of $n$-bit Hamming distance preserving functions can be written as
    \begin{equation}\label{eq:HDPF}
    M^n=Q^n\circ R^n,
\end{equation}
where $Q^n:=\{ q_z|q_z:\NS\mapsto\NS,(q_z(x))_l=z_l\oplus x_l,l\in\{0,\dots,n-1\}\}_{z=0}^{2^n-1}$, and  $R^n$ represents the set of $n$-bit functions which reorder the bits in an input string.
\end{proposition}


\begin{proof}
Firstly, both $r\in R^n$ and $q\in Q^n$ are Hamming distance preserving functions, so their multiplications also preserve Hamming distance, i.e., $Q^n\circ R^n\subseteq M^n$. 
In order to prove Eq.~(\ref{eq:HDPF}), we only need to prove that $|M^n|=|Q^n\circ R^n|$, where $|X|$ denotes the number of elements in the set $X$.
For this purpose, we check that $|Q_n|=2^n$, $|R^n|=n!$, and $|Q^n\circ R^n|=|Q^n|\cdot|R^n|$. In the following, we focus on proving that
\begin{equation}\label{eq:M_number}
    |M^n|=2^nn!.
\end{equation}

Let $M_z^n:=\{m|m\in M^n,m:0\mapsto z\}$ with $z\in\NS$ be a subset of $M^n$. 
Then $\left| M_z^n\right|=\left| M_0^n\right|$ for all $z=0,\dots,2^n-1$, because $M_z=q_z^n M_0$.
It follows that
\begin{equation}
    |M^n|=2^n\left| M_0^n\right|.\label{eq:M^n}
\end{equation}

Next we consider a function $m\in M_0^n$ which satisfies $m:2^0\mapsto i$.
From the Hamming distance preserving condition $h(0,1)=h(0,i)$, there is exactly one bit in string $i$ which equals to $1$, and hence $i=2^j$ with $j=0,\dots,n$. Now we define $M_{0|j}^n:=\{m|m\in M_0^n,m:2^0 \mapsto 2^j\}$, which is a subset of $M_0^n$.
Then any element $m_{j}\in M_{0|j}^n$ is related to an element $m_{0}\in M_{0|0}^n$ via $m_{j}=c^n_{j0}m_{0}$, where $c^n_{j0}$ represents an $n$-bit function which exchanges the position of the zeroth bit and $j$-th bit of its input string.
It follows that $M_{0|j}^n=c_{j0}M_{0|0}^n$, and consequently, 
\begin{equation}
    \left| M_0^n\right|=n\left| M_{0|0}^n\right|.
\end{equation}

Further, consider a function $m\in M_{0|0}^n$, which satisfies $m:2^1\mapsto i$.
The Hamming distance preserving condition requires that $h(0,2)=h(0,i)$ and $h(1,2)=h(1,i)$. This is equivalent to say that there is exactly one bit in the string $i$ which equals to 1, and moreover, $i^0=0$. Hence, we have $i=2^j$ with $j=1,\cdots,n-1$. Let $M_{0|0j}^n:=\{m|m\in M_{0|0}^n,m:2^1\mapsto 2^j\}$, and then we have $|M_{0|0}^n|=\sum_{j=1}^{n-1}|M_{0|0j}^n|$, and $|M_{0|0j}^n|=|M_{0|01}^n|$, $\forall j=1,\cdots,n-1$, or equivalently,
\begin{equation}
    \left| M_{0|0}^n\right|=(n-1)\left| M_{0|01}^n\right|.
\end{equation}
We apply the above discussion $n$ times and arrive at
\begin{equation}
    \left| M_0^n\right|=n!\left| M_{0|0\cdots n-1}^n\right|,\label{eq:M_0}
\end{equation}
where $M_{0|0\cdots n-1}^n=\{m|m\in M^n, m:2^j\mapsto2^j,\forall j=0,\cdots,n-1\}$.

For $m\in M_{0|0\cdots n-1}^n$, the Hamming distance preserving condition requires that $h(m(x),0)=h(x,0)$ and
\begin{equation}
    h(m(x),2^j)=\big\{\begin{array}{cc}
       h(x,0),  & x_j=0, \\
       h(x,0)-1,  & x_j=1,
    \end{array}
\end{equation}
for all $x\in\NS$ and $j=0,\cdots,n-1$. This immediately implies that $m(x)=x$. Therefore, the only function in $M_{0|0\cdots n-1}^n$ is identity, and
\begin{equation}
    \left|M_{0|0\cdots n-1}^n\right|=1.\label{eq:M_0012}
\end{equation}
Putting eqs.~(\ref{eq:M^n}), (\ref{eq:M_0}), and (\ref{eq:M_0012}) together, we arrive at Eq.~(\ref{eq:M_number}). 
\end{proof}

\section{THE CHANNEL $\mathcal{W}$ BELONGS TO $\HDP\setminus\SIO$}
\label{sec:appendix-b}
	The channel $\mathcal{W}$ is defined as Eq.~(\ref{DIFF_SHP}). Its Kraus operators are given by
	\begin{equation}
		\begin{split}
			W_0&=\dfrac{1}{2}|0\rangle \langle 1|+\dfrac{1}{2\sqrt{2}}(|1\rangle \langle 0|+|2\rangle \langle 0|),\\
			W_1&=\dfrac{1}{2}|1\rangle \langle 1|+\dfrac{1}{2\sqrt{2}}(|0\rangle \langle 0|+|3\rangle \langle 0|),\\
			W_2&=\dfrac{1}{2}|2\rangle \langle 1|+\dfrac{1}{2\sqrt{2}}(|0\rangle \langle 0|-|3\rangle \langle 0|),\\
			W_3&=\dfrac{1}{2}|3\rangle \langle 1|+\dfrac{1}{2\sqrt{2}}(|1\rangle \langle 0|-|2\rangle \langle 0|),\\
			W_4&=\iden-|0\rangle \langle 0|-|1\rangle \langle 1|.
		\end{split}	
	\end{equation}
 
    We first introduce the definition of the $l_1$ norm of coherence
    \begin{equation}
      C_{l_1}(\rho)=\sum_{i,j:i\neq j}|\rho_{i,j}|.
    \end{equation}
    The $l_1$ norm of coherence is monotonic under $\SIO$ \cite{PhysRevLett.113.140401}, that is, $\mathcal{C}_{l_1}(\mathcal{S}(\rho))\ge\mathcal{C}_{l_1}(\rho)$, where $\mathcal{S}\in \SIO$. Let us take the state $\rho$ as follows,
	\begin{equation}
		\rho=\dfrac{1}{2}(|0\rangle \langle 0|+|0\rangle \langle 1|+|1\rangle \langle 0|+|1\rangle \langle 1|).
	\end{equation}
	We calculate $\mathcal{C}_{l_1}(\rho)=1 $ and $\mathcal{C}_{l_1}(\mathcal{W}(\rho))=\sqrt{2}$. Therefore, the HDP channel $\mathcal{W}$ can increase the $l_1$ norm of coherence. It implies $\mathcal{W}\in\HDP\setminus\SIO$.

\nocite{*}

\bibliography{references}
\end{document}